\documentclass[a4paper,11pt]{amsart}

\usepackage[left=2.5cm,right=2.5cm,top=3cm,bottom=3cm]{geometry}

\usepackage{amsmath}
\usepackage{amssymb}
\usepackage{amsthm}
\usepackage{pgf,tikz}
\usepackage{graphicx}
\usepackage{setspace}
\usepackage[version=3]{mhchem}
\usepackage[sort&compress,comma,square,numbers]{natbib}
\usepackage{enumerate}

\usepackage{url}

\usepackage{mathabx}
 
\newcommand{\R}{\mathbb{R}}

\newcommand{\kk}{{\kappa}}

\newtheorem{theorem}{Theorem}

\theoremstyle{definition}

\begin{document}

  \begin{abstract}
Parametrized polynomial ordinary differential equation systems are broadly used for modeling, specially in the study of biochemical reaction networks under the assumption of mass-action kinetics. Understanding the qualitative behavior of the solutions with respect to the parameter values gives rise to complex problems within real algebraic geometry, concerning the study of the signs of multivariate polynomials over the positive orthant. In this work we provide further insight into the number of positive steady states  of a benchmark model, namely the two-site phosphorylation cycle. In particular, we provide new conditions on the reaction rate constants for the existence of one or three positive steady states, partially filling in a gap left in previous works.
\end{abstract}

\title[Reaction rate constants that enable multistationarity]{On the reaction rate constants that enable multistationarity in the two-site phosphorylation cycle}

\author[E. Feliu]{Elisenda Feliu$^{1}$}
\date{\today}

\footnotetext[1]{Department of Mathematical Sciences, University of Copenhagen, Universitetsparken 5, 2100 Copenhagen, Denmark. efeliu@math.ku.dk}

 \tikzset{every node/.style={auto}}
 \tikzset{every state/.style={rectangle, minimum size=0pt, draw=none, font=\normalsize}}
 \tikzset{bend angle=15}

\maketitle

 \section{Introduction}
Biologists use model organisms such as bacteria \textit{Escherichia coli}, the fruit fly \textit{Drosophila melanogaster}, or the mouse \textit{Mus musculus}    to gain insight about particular phenomena, with the hope that these can be taken as representatives of how other closely related organisms possibly behave. In the study of models of biochemical reaction networks, a specific reaction network has become the benchmark model where new techniques, strategies and approaches are tested. Hence this model has become the \emph{model model}, and we expect that the strategies employed to answer mathematical questions about it can be used to approach similar systems arising in molecular biology. 

The model model is a simple model of phosphorylation and dephosphorylation processes. These processes are central in the modulation of cell communication, activities and responses. For example, phosphorylation affects about $30\%$ of all  proteins in humans \cite{cohen}.  
The reaction network we consider consists of a substrate $S$ that has two phosphorylation sites. Phosphorylation occurs distributively in an ordered manner, such that one of the sites is always phosphorylated first. We denote with $S_0,S_1,S_2$ the three phosphoforms of $S$ with $0,1,2$ phosphorylated sites, respectively, and assume that a kinase $E$ and a phosphatase $F$ mediate the phosphorylation  and dephosphorylation of $S$ respectively.
This gives rise to the following reactions \cite{Wang:2008dc,conradi-mincheva}:
\begin{align}\label{eq:network}
\begin{split}
S_0 + E \ce{<=>[\kk_1][\kk_2]} ES_0 \ce{->[\kk_3]} S_1+E   \ce{<=>[\kk_7][\kk_8]} ES_1 \ce{->[\kk_9]} S_2+E \\
S_2 + F  \ce{<=>[\kk_{10}][\kk_{11}]} FS_2 \ce{->[\kk_{12}]} S_1+F  \ce{<=>[\kk_4][\kk_5]} FS_1 \ce{->[\kk_6]} S_0+F.
\end{split}
\end{align}
Under the assumption of mass-action kinetics, the evolution of the concentration of the species of the network in time is modeled by a system of autonomous ordinary differential equations (ODEs) in $\R^9_{\geq 0}$. The system consists of polynomial equations, whose coefficients are scalar multiples of one of  $12$ positive parameters $\kk_1,\dots,\kk_{12}$. Further, the dynamics are constrained to   linear invariant subspaces of dimension six, characterized  by the total amounts of kinase, phosphatase and substrate, which so enter the study as parameters.

This system is large enough for hands-on approaches to fail, but small enough to challenge the development of new mathematics. 
Further, dynamical properties of the ODE system of this network might be lifted to more complex networks related to it. For example, \eqref{eq:network} is an example of an $n$-site phosphorylation cycle \cite{Wang:2008dc,TG-Nature,FHC14}, a post-translational modification network \cite{TG-rational,fwptm,conradi-shiu-review}, a MESSI system \cite{Dickenstein-MESSI}, a network with toric steady states \cite{PerezMillan}  and, importantly, it is a building block of the MAPK cascade involved ubiquitously in  cell signaling  \cite{Huang-Ferrell,qiao:oscillations,rendall-MAPK}. 

Currently, it is known 
that the number of positive steady states within the linear invariant subspaces is either one or three, if all positive steady states are  non-degenerate \cite{Wang:2008dc,Markevich-mapk}. It has also been shown that there are choices of parameters for which there are two asymptotically stable steady states and one unstable steady state \cite{rendall-2site}, see also \cite{torres:stability}.

Some recent progress has shed some light on how these qualitative properties depend on the choice of parameters. In \cite{conradi-mincheva} the authors give two rational functions on the parameters  $\kk_1,\dots,\kk_{12}$, say $a(\kk)$ and $b(\kk)$, such that the system has one positive steady state in each invariant linear subspace if  $a(\kk)\geq 0$ and $b(\kk)\geq 0$, and has at least two in some invariant linear  subspace if $a(\kk)<0$. Furthermore in \cite{Feliu:royal,dickenstein:regions} conditions for the existence of three positive steady states involving the parameters  $\kk_1,\dots,\kk_{12}$ and some of the total amounts are given, see also \cite{conradi-erk}. Nevertheless, our understanding of the parameter region with at least two positive steady states is incomplete, and questions such as \emph{is it connected?} are still open. 

The difficulties in understanding the number of steady states (and also their stability or existence of Hopf bifurcations) arise from the high number of parameters and variables combined with the difficulties in studying polynomials over the positive real numbers. 
This is what left  the scenario $a(\kk)\geq 0$ and $b(\kk)<0$ open in \cite{conradi-mincheva}. In this work, we cover this case by showing that there are open sets in the parameters  $\kk_1,\dots,\kk_{12}$ satisfying these two inequalities both for the existence of only one positive steady state in each invariant linear subspace, and the existence of at least two in some linear invariant subspace.

\section{The parametric polynomial system}\label{sec:theproblem}
We consider the reaction network \eqref{eq:network} and 
  denote the concentrations of the species by
$x_1=[E], x_2=[F]$, $x_3=[S_0]$, $x_4=[S_1]$, $x_5=[S_2]$, $x_6=[ES_0]$, $x_7=[FS_1]$, $x_8=[ES_1]$, $x_9=[FS_2]$. Under mass-action kinetics, the ODE system modelling the concentrations of the nine species in the network \eqref{eq:network}  over time $t$  is
{\small \begin{align}
\tfrac{dx_1}{dt} &= -\kk_{1}x_{1}x_{3}-\kk_{7}x_{1}x_{4}+\kk_{2}x_{6}+\kk_{3}x_{6}+\kk_{8}x_{8}+\kk_{9}x_{8} & \tfrac{dx_6}{dt} &= \kk_{1}x_{1}x_{3}-\kk_{2}x_{6}-\kk_{3}x_{6} \nonumber
\\ 
\tfrac{dx_2}{dt} &= -\kk_{4}x_{2}x_{4}-\kk_{10}x_{2}x_{5}+\kk_{5}x_{7}+\kk_{6}x_{7}+\kk_{11}x_{9}+\kk_{12}x_{9} & \tfrac{dx_7}{dt} &= \kk_{4}x_{2}x_{4}-\kk_{5}x_{7}-\kk_{6}x_{7} \nonumber
\\ 
\tfrac{dx_3}{dt} &=-\kk_{1}x_{1}x_{3}+\kk_{2}x_{6}+\kk_{6}x_{7} & \tfrac{dx_8}{dt} &= \kk_{7}x_{1}x_{4}-\kk_{8}x_{8}-\kk_{9}x_{8}  \label{eq:ode} \\
\tfrac{dx_4}{dt} &= -\kk_{4}x_{2}x_{4}-\kk_{7}x_{1}x_{4}+\kk_{3}x_{6}+\kk_{5}x_{7}+\kk_{8}x_{8}+\kk_{12}x_{9} & \tfrac{dx_9}{dt} &= \kk_{10}x_{2}x_{5}-\kk_{11}x_{9}-\kk_{12}x_{9}.   \nonumber
\\
\tfrac{dx_5}{dt} &=-\kk_{10}x_{2}x_{5}+\kk_{9}x_{8}+\kk_{11}x_{9} \nonumber
\end {align}}%
The dynamics of this system take place in linear invariant subspaces of dimension six, defined by the equations
\begin{equation}\label{eq:cons_laws}
 x_1+ x_6 + x_8 = E_{\rm tot},\quad x_2+x_7+x_9=F_{\rm tot},\quad  x_3+x_4+x_5+x_6+x_7+x_8+x_9= S_{\rm tot},
\end{equation}
subject to $x_i\geq 0$ for $i=1,\dots,9$. Here $E_{\rm tot}, F_{\rm tot}, S_{\rm tot}$ stand for the total amounts of kinase $E$, phosphatase $F$ and substrate $S$.
The steady states of the network within one of these linear invariant subspaces are the solutions to the  system of polynomial equations given by  $\tfrac{dx_i}{dt}=0$, for $i=1,\dots,9$ in \eqref{eq:ode}, and the equations in \eqref{eq:cons_laws}. This system is called the \emph{steady state system} (note that three of the equations $\tfrac{dx_i}{dt}=0$ are redundant, say, the ones for $x_1,x_2,x_3$, and can be removed).

The steady state system depends on the variables $x_1,\dots,x_{9}$ and the parameters $\kk_1,\dots,\kk_{12}$, $E_{\rm tot}, F_{\rm tot}, S_{\rm tot}$, all of which are assumed to be positive. 
We say that a given vector of reaction rate constants  $\kk=(\kk_1,\dots,\kk_{12})$ \emph{enables} multistationarity, if there exist $E_{\rm tot}, F_{\rm tot}, S_{\rm tot}$  such that the steady state system  admits at least two positive solutions, that is, with all coordinates positive. We say then that the network is multistationary in the linear invariant subspace with total amounts  $E_{\rm tot}, F_{\rm tot}, S_{\rm tot}$.

In \cite{conradi-mincheva}, see also \cite{FeliuPlos}, conditions on the reaction rate constants for enabling, or not, multistationarity where given. 
Specifically, 
consider the Michaelis-Menten constants of each phosphorylation/dephosphorylation event:
\[K_1 = \tfrac{\kk_2+\kk_3}{\kk_1},\quad  K_2 = \tfrac{\kk_5+\kk_6}{\kk_4},\quad K_3 = \tfrac{\kk_8+\kk_9}{\kk_7},\quad K_4 = \tfrac{\kk_{11}+\kk_{12}}{\kk_{10}}.\]
The map
$\pi\colon \R^{12}_{>0}  \rightarrow  \R^{8}_{>0}$ sending $\kk=(\kk_1,\dots,\kk_{12})$ to   $\eta=(K_1,K_2,K_3,K_4,\kk_3,\kk_6,\kk_9,\kk_{12})$
is continuous and surjective. 
By letting
\[ a(\eta)= \kk_3\kk_{12}-\kk_6\kk_9, \qquad b(\eta)=(K_{{2}}+K_{{3}})\kk_{{3}}\kk_{{12}}-(K_1+K_{{4}})\kk_{{6}}\kk_{{9}},
 \]
we have that 
\begin{itemize}
\item if $a(\eta)<0$, then any $\kk\in \pi^{-1}(\eta)$ enables multistationarity;
\item if $a(\eta)\geq 0$ and $b(\eta)\geq 0$, then the steady state system has  exactly one positive solution for all $E_{\rm tot}, F_{\rm tot}, S_{\rm tot}$ (hence a $\kk\in \pi^{-1}(\eta)$ does not enable  multistationarity).
\end{itemize}
 In \cite{conradi-mincheva,FeliuPlos}  the scenario  $a(\eta)\geq 0$ and $b(\eta)< 0$ was left open.
We show now that in this case, there exist choices of reaction rate constants that enable multistationarity and choices that do not.

The inequalities above are found by combining Brouwer degree theory with some algebraic manipulations. This approach was first introduced in \cite{conradi-mincheva} for this specific system, and was further adapted to generic systems satisfying certain technical conditions in \cite{FeliuPlos}. We follow the presentation of the latter here. 
Consider the polynomial in $x_1,x_2,x_3$ with coefficients depending on $\eta$:
{\small \begin{align}\label{eq:mypolynomial}
\begin{split}
p_\eta(x)& = K_{2}^{2}K_{4}  \kk_{3}^{2}\kk_{9} ( \kk_{3}\kk_{12}-\kk_{6}\kk_{9}) x_{1}^{4}x_{3}^{2}+K_{1} K_{2}^{2}K_{4}\kk_{3}^{2}\kk_{6}\kk_{9}^{2}\, x_{1}^{4}x_{3}\\ &  + K_{1}K_{2}K_{3} \kk_{3}\kk_{6} \kk_{12}  (\kk_{3}\kk_{12}-\kk_{6}\kk_{9}) x_{1}^{3}x_{2}^{2}x_{3}
+K_{2}^{2}K_{3} \kk_{3}^{2}\kk_{12} ( \kk_{3}\kk_{12}- \kk_{6}\kk_{9}) x_{1}^{3}x_{2}x_{3}^{2} \\ &
+2K_{1}K_{2}K_{3}K_{4}\kk_{3}^{2}\kk_{6}\kk_{9}\kk_{12} \, x_{1}^{3}x_{2}x_{3} 
+ K_{1}K_{2}K_{3}\kk_{3}\kk_{6}\kk_{12} ( \kk_{3}\kk_{12}-\kk_{6}\kk_{9} ) x_{1}^{2}x_{2}^{3}x_{3} 
\\ & +K_{1}^{2}K_{2}K_{3}\kk_{3}\kk_{6}^{2} \kk_{12}( \kk_{9}+\kk_{12}) x_{1}^{2}x_{2}^{3} 
+K_{1}K_{2}K_{3}\kk_{3}\kk_{6}\kk_{12}  (\kk_{3}\kk_{12}-\kk_{6}\kk_{9}) x_{1}^{2}x_{2}^{2}x_{3}^{2}
\\ &  +K_{1}K_{2}K_{3}\kk_{3}\kk_{6}\kk_{12} ( (K_{2}+K_3)\kk_{3}\kk_{12}-(K_{1}+K_4)\kk_{6}\kk_{9}) x_{1}^{2}x_{2}^{2}x_{3}
\\ &
+K_{1}^{2}K_{2}K_{3}K_{4}\kk_{3}\kk_{6}^{2}\kk_{9}\kk_{12}\, x_{1}^{2}x_{2}^{2}
 +K_{1}^{2}K_{3}^{2}\kk_{6}^{2}\kk_{12}^{2} (\kk_{3}+\kk_{6})\,  x_{1}x_{2}^{4} 
+2\,K_{1}^{2}K_{2}K_{3}\kk_{3}\kk_{6}^{2}\kk_{12}^{2}\, x_{1}x_{2}^{3}x_{3} \\ &+K_{1}^{2}K_{2}K_{3}^{2}\kk_{3}\kk_{6}^{2}\kk_{12}^{2}\, x_{1}x_{2}^{3}
 +K_{1}^{2}K_{3}^{2}\kk_{6}^{3}\kk_{12}^{2}\, x_{2}^{4}x_{3}+K_{1}^{3}K_{3}^{2}\kk_{6}^{3}\kk_{12}^{2}\,  x_{2}^{4}.
\end{split}
\end{align}}Then 
\begin{itemize}
\item[(Mono)]  If $p_\eta(x)$ is positive for all $x_1,x_2,x_3>0$, then $\kk\in \pi^{-1}(\eta)$ does not enable multistationarity, and further there is exactly one positive steady state in each invariant linear subspace.
\item[(Mult)] If $p_\eta(x)$  is negative for some $x_1,x_2,x_3>0$, then $\kk\in \pi^{-1}(\eta)$ enables multistationarity in the invariant linear subspace containing the point
{\small \begin{align}\label{eq:parametrization}
\varphi(x_1,x_2,x_3) & = \left(x_1,x_2,x_3, \frac{K_2\kk_3 x_1x_3}{K_1\kk_6 x_2}, \frac{K_2K_4\kk_3\kk_9x_1^{2} x_3}{K_1K_3\kk_6\kk_{12}x_2^{2}},
\frac{x_1x_3}{K_1}, \frac{\kk_3x_1x_3}{K_1\kk_6},  \frac{K_2\kk_3x_1^{2}x_3}{K_1K_3\kk_6 x_2}, \frac{K_2 K_3\kk_3 \kk_9 x_1^{2} x_3}{K_1\kk_6\kk_{12}x_2}\right).
\end{align}}
\end{itemize}

The coefficients of $p_\eta(x)$ in $x=(x_1,x_2,x_3)$ are polynomials in the eight parameters $ K_1,K_2$, $K_3$, $K_4$, $\kk_3$, $\kk_6$, $\kk_9,\kk_{12}$, which either have positive coefficients, or are multiples of $a(\eta)$ or $b(\eta)$.
If $a(\eta)\geq 0$ and $b(\eta) \geq 0$, then $p_\eta(x)$ is a polynomial in $x$ with all coefficients positive, and hence becomes positive when evaluated at any positive vector $x$. By (Mono),   there is one positive steady state in each invariant linear subspace defined by the equations \eqref{eq:cons_laws}.

Now, assume $a(\eta)<0$. Then by letting $x_1=y, x_2=\tfrac{1}{y}, x_3=y$, $p_\eta(x)$ becomes a polynomial in one variable, $y$, with leading term  $K_{2}^{2}K_{4}  \kk_{3}^{2}\kk_{9} a(\eta)<0$. Hence  $p_\eta(x)$ is negative for $y>0$ large enough.  By (Mult), we conclude that  multistationarity is enabled for $\kk\in \pi^{-1}(\eta)$.

If $b(\eta)<0$ and $a(\eta)\geq 0$, then $p_\eta(x)$ has only one negative coefficient, corresponding to the monomial $x_1^2 x_2^2 x_3$. The same trick as above fails because the monomial  $x_1^2 x_2^2 x_3$ is not a vertex of the \emph{Newton polytope} of $p_\eta(x)$ (the convex hull of the exponent vectors for all monomials in $x$ of the polynomial $p_\eta(x)$). Specifically, it is well known that if $\alpha$ is a vertex of the Newton polytope $\mathcal{N}$ of a multivariate polynomial $q(x)= a x^\alpha + \dots \in \R[x_1,\dots,x_n]$, then there exist positive vectors $x$ such that the sign of $q(x)$ agrees with the sign of $a$, independently of the value of $a$. In  order to construct such a point $x$, one considers a separating hyperplane containing the vertex $\alpha$ and such that all other vertices are on the same side of the hyperplane. In particular, such an  hyperplane can take the form $\omega\cdot (x-\alpha)=0$ with $\omega$ such that $\omega\cdot \beta<\omega\cdot \alpha$ for all other vertices of $\mathcal{N}$. Then, for $x=\prod_{i=1}^n y^{\omega_i}$,   the sign of $q(x)$ is the sign of $a$ for $y>0$ large enough.

The Newton polytope of the polynomial \eqref{eq:mypolynomial} is shown in Figure~\ref{fig:newton}, seen from two different angles. The blue point on the right display is $(2,2,1)$, which corresponds to the monomial  $x_1^2 x_2^2 x_3$. Clearly, this point is not a vertex of the polytope, and hence this strategy is not informative when $b(\eta)<0$. The point $(4,0,2)$, corresponding to the monomial $x_1^4x_3^2$, is a vertex.

\begin{figure}[t]
\begin{center}
\begin{minipage}[h]{0.4\textwidth}
\includegraphics[scale=0.4]{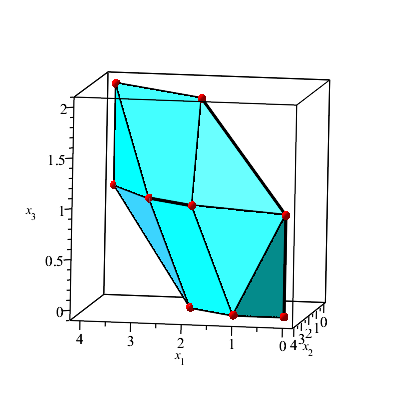}
\end{minipage}
\quad
\begin{minipage}[h]{0.4\textwidth}
\includegraphics[scale=0.4]{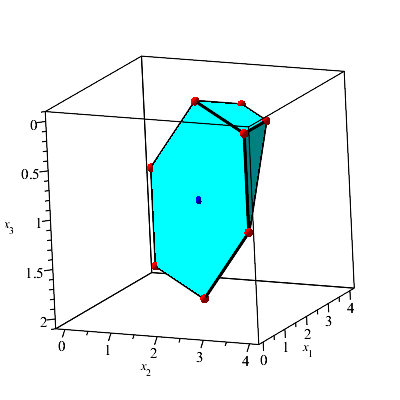}
\end{minipage}
\end{center}
\caption{{\small Newton polytope of the polynomial in \eqref{eq:mypolynomial}.}} \label{fig:newton}

\end{figure}

\section{New regions of multistationarity}
Here we show  there exists a nonempty open set $\widehat{U}$ in $\R^{12}_{>0}$ such that all $\kk\in \widehat{U}$ satisfy that $b(\eta)<0$, $a(\eta)\geq 0$ and enable multistationarity, and similarly, that there exists another nonempty open set $\widehat{V}$   in $\R^{12}_{>0}$ such that all $\kk\in \widehat{V}$ also satisfy $b(\eta)<0$ and $a(\eta)\geq 0$, but do not enable multistationarity.
The statement is a direct consequence of the following theorem, using that $\pi$ is continuous. By abuse of notation we write $a(\eta')=\kk_3\kk_{12}-\kk_6\kk_9$   for $\eta'=(K_2,K_3,K_4,\kk_3,\kk_6,\kk_9,\kk_{12})$.

\begin{theorem}\label{thm:main}
There exist two nonempty open sets $U,V\subseteq \R^8_{>0}$ such that $a(\eta)\geq 0$ and $b(\eta)<0$ for all $\eta\in U\cup V$,   and further
$p_\eta(x)$ is negative for some $x\in \R^3_{>0}$ if $\eta\in U$, and $p_\eta(x)$ is positive for all $x\in \R^3_{>0}$ if $\eta\in V$.
Furthermore, we have
\begin{enumerate}[(i)]
\item If $a(\eta)\geq 0$ and 
\[  -b(\eta)  < \max\big(2\sqrt{K_1K_4}\kk_{6}\kk_{9},2 \sqrt{a(\eta) K_2 K_3  \kk_3 \kk_{12}}\, \big) ,\]
then all $\kk\in \pi^{-1}(\eta)$ do not enable multistationarity.
\item For any choice of $\eta'=(K_2,K_3,K_4,\kk_3,\kk_6,\kk_9,\kk_{12})\in \R^7_{>0}$, multistationarity is enabled for $K_1$   larger than $y_0^3$, where $y_0$  is the largest real root of the following polynomial in $y$ (if any):
{\small 
\begin{align*}
g_{\eta'}(y) &= -K_{2}K_{3}\kk_{3}\kk_6^{2}\kk_9\kk_{12}\,  y^{3}+
\kk_6 \big( K_{2}^{2}K_{4}\kk_{3}^{2}\kk_9^{2}+K_{2}K_{3}  \kk_{3}\kk_{12}  (a(\eta'')+2  \kk_6\kk_{12})+K_{3}^{2}\kk_6^{2}\kk_{12}^{2} \big) y^{2} \\
& + \big( K_{2}^{2}K_{4}\kk_3^2 \kk_9 a(\eta')+ K_{2}K_{3}K_{4} \kk_{3}\kk_6\kk_9\kk_{12} ( 2\kk_{3} + \kk_6) +K_{2}K_{3}  \kk_3 \kk_6\kk_{12}^{2} ( \kk_{3}+\kk_6)+K_{3}^{2}\kk_6^{3}\kk_{12}^{2} \big) y\\ 
& +K_{3}\kk_{12} \big( K_{2}^{2}\kk_{3}^{2}  (a(\eta')+\kk_6\kk_{12})+ K_{3} \kk_6\kk_{12} ( K_{2}\kk_{3} + \kk_6)  (\kk_{3}+\kk_6) -K_{2}K_{4}\kk_{3}\kk_6^{2}\kk_9+K_{2} \kk_{3}\kk_6 a(\eta \big).
\end{align*}
}
Analogously, given $K_1,K_2,K_3,\kk_3,\kk_6,\kk_9,\kk_{12}$,  multistationarity is enabled for $K_4$ larger than $y_0^3$, where $y_0$ is  the largest real root of $g_{\sigma(\eta')}(y)$ for  $\sigma(\eta')=(K_3,K_2,K_1,\kk_{12},\kk_9,\kk_6,\kk_{3})$.
\end{enumerate}
\end{theorem}
 \begin{proof}
Viewed as a polynomial in $K_1,x_1,x_2,x_3$,  $p_\eta(x)$ has the following  monomials:
{\small \begin{align*}
K_{1}^{2} x_{2}^{4}x_{3},\ K_{1}^{2}x_{1}^{2}x_{2}^{2},\ K_{1}^{2} x_{1}x_{2}^{3}x_{3},\ K_{1}^{2}x_{1}^{2}x_{2}^{2}x_{3},\ K_{1}x_{1}^{3}x_{2}^{2}x_{3},\ K_{1}x_{1}^{2}x_{2}^{2}x_{3}^{2},\ K_{1}x_{1}^{2}x_{2}^{2}x_{3}, \\ \ K_{1}x_{1}^{2}x_{2}^{3}x_{3},  \ K_{1} x_{1}^{3}x_{2}x_{3},\ K_{1}^{3} x_{2}^{4},\ K_{1}x_{1}^{4}x_{3},\ K_{1}^{2}x_{1}^{2}x_{2}^{3},\ K_{1}^{2} x_{1}x_{2}^{3},\ x_{1}^{4} x_{3}^{2},\ x_{1}^{3}x_{2} x_{3}^{2} ,\ K_{1}^{2}x_{1}x_{2}^{4}.
\end{align*}}%
The coefficient  of $K_1^2  x_1^2 x_2^2 x_3$ is $- K_2K_3 \kk_3\kk_6^2\kk_9\kk_{12}$, which is negative. The exponent vector $(2,2,2,1)$ of this monomial is a vertex of the Newton polytope associated with the polynomial $p_\eta(x)$ in $\R[K_1,x_1,x_2,x_3]$. A separating hyperplane is defined by the vector
$\omega = (3,1,0,2)$. Following the procedure explained at the end of Section~\ref{sec:theproblem}, we consider now 
$K_1= y^{3}$, $x_1=y$, $x_2=1$, $x_3=y^2$. By letting $\eta'=(K_2,K_3,K_4,\kk_3,\kk_6,\kk_9,\kk_{12})$,  $p_\eta(x)$ becomes $y^{7}g_{\eta'}(y)$.
Since the leading term of $g_{\eta'} (y)$  is negative,  $g_{\eta'} (y)$ is negative for some value of $y$,  and this is guaranteed at least if $y$ is large enough. There could be other values of $y$ for which $g_{\eta'}(y)$ is negative, but this depends on whether  $g_{\eta'} (y)$ has two, one or zero positive real roots.
Any $y$ for which  $g_{\eta'}(y)<0$  defines a value of $K_1$ and $x\in \R^3_{>0}$ such that $p_\eta(x)$  is negative.  
Specifically, consider any $y$ larger than the largest real root of $g_{\eta'} (y)$, if any. Since $K_1 =y^{3}$, increasing $y$ implies making   $K_1$ large.
 In this case necessarily $b(\eta)<0$, since $a(\eta)\geq 0$ by assumption and $a(\eta)$ does not depend on $K_1$. 
 
Let $U'$ be the open set of parameters $K_2,K_3,K_4,\kk_3,\kk_6,\kk_9,\kk_{12}$ in  $\R^7_{>0}$ such that $a(\eta)> 0$, and consider the   polynomial $g_{\eta'}(y)$ for all $\eta'\in U'$. By continuity of the roots of a polynomial, there exists a continuous and positive function $\epsilon(\eta')$ defined on $U'$ such that if  $y>\epsilon(\eta')$, then $g_{\eta'}(y)$ is negative and $b(\eta)<0$ with $K_1=y^3$. 
 It follows that for all $\eta$ in the open set
\[ U=\big\{ (y^3,  K_2,K_3,K_4,\kk_3,\kk_6,\kk_9,\kk_{12}) \mid \kk_3\kk_{12} - \kk_6\kk_9>0, y>\epsilon(\eta')\big\},\] $p_\eta(x)$ is negative for some $x\in \R^3_{>0}$. 
This shows the first part of the theorem and statement (ii) for $K_1$. Statement (ii) on $K_4$ follows analogously 
   because the system is symmetric under the change of variables  sending 
$x\in \R^9$ to $(x_2,x_1,x_5,x_4,x_3,x_9,x_8,x_7,x_6)$
and  the parameter $\kk$ to $(\kk_{10},$ $ \kk_{11},\kk_{12},\kk_7,\kk_8,\kk_9,\kk_4,\kk_5,\kk_6,\kk_1,\kk_2,\kk_3)$. 

\medskip
To show the second part, 
we identify a square polynomial that involves $b(\eta)$. Specifically, 
  the polynomial $p_\eta(x)$, seen as a polynomial in $x$,  has the  summand $K_{1} \kk_{6}x_3 \alpha_\eta(x)$ with
{\small \[ \alpha_\eta(x) = K_{2}^{2}K_{4}\kk_{3}^{2}\kk_{9}^{2}x_{1}^{4} +K_{2}K_{3} \kk_{3}\kk_{12}b(\eta) x_{1}^{2}x_{2}^{2} +K_1K_{3}^{2}\kk_{6}^{2}\kk_{12}^{2}x_{2}^{4}.\]}%
The following polynomial  is positive for all real values of $\eta$ and $x$:
{\small \begin{multline*}
\big( \sqrt{K_4}K_{2}\kk_{3}\kk_{9}x_{1}^{2} -\sqrt{K_1}K_{3}\kk_{6}\kk_{12}x_{2}^{2}\big)^2   =K_{2}^{2}K_{4}\kk_{3}^{2}\kk_{9}^{2}x_{1}^{4}   -  2\sqrt{K_1K_4}K_{2}K_{3}\kk_{3} \kk_{6}\kk_{9}\kk_{12}x_{1}^{2}x_{2}^{2}   +  K_1K_{3}^{2}\kk_{6}^{2}\kk_{12}^{2}x_{2}^{4}.
\end{multline*}}%
If $ K_{2}K_{3} \kk_{3}\kk_{12}\, b(\eta)>-2\sqrt{K_1K_4}K_{2}K_{3}\kk_{3} \kk_{6}\kk_{9}\kk_{12}$, 
or equivalently,
\begin{equation}\label{eq:myeq1}
-b(\eta) < 2\sqrt{K_1K_4} \kk_{6}\kk_{9},
\end{equation}
then $\alpha_\eta(x)>  \big( \sqrt{K_4}K_{2}\kk_{3}\kk_{9}x_{1}^{2} -\sqrt{K_1}K_{3}\kk_{6}\kk_{12}x_{2}^{2}\big)^2>0$, and hence $p_\eta(x)>0$ for all $x\in \R^3_{>0}$.
Similarly,  $ K_2K_3\kk_3\kk_{12} x_1x_2 \alpha_\eta(x)$ with
{\small
 \[ \alpha_\eta(x) =K_1^2 K_3 \kk_6^2\kk_{12} x_2^2  +    K_1\kk_6 b(\eta) x_1 x_2 x_3 + a(\eta) K_2 \kk_3 x_1^2 x_3^2,\]}%
 is a summand of $p_\eta(x)$. By expanding 
$ ( K_1\kk_6 \sqrt{K_3 \kk_{12}}x_2  - \sqrt{a(\eta) K_2 \kk_3} x_1 x_3)^2$ we conclude as above that $p_\eta(x)>0$ for all $x\in \R^3_{>0}$ if 
$ K_1\kk_6 b(\eta) > -2  K_1 \kk_6 \sqrt{a(\eta) K_2K_3 \kk_3\kk_{12}}$
or equivalently
\begin{equation}\label{eq:myeq2}
- b(\eta) < 2 \sqrt{a(\eta) K_2 K_3  \kk_3 \kk_{12}}.
\end{equation}
By combining \eqref{eq:myeq1} and \eqref{eq:myeq2} we obtain statement (i).
Consider the open set $V$ consisting of $\eta\in \R^8_{>0}$ such that $a(\eta)>0$ and \eqref{eq:myeq1} or \eqref{eq:myeq2} hold.    Then  $V$ satisfies the assumptions in the statement since it is nonempty as $\eta=(2,0.5,0.5,2,2,1,1,1)\in V$.
This concludes the proof.
\end{proof}
 
 The statement  of the theorem gives us a way to construct sets of reaction rate constants such that $a(\eta)\geq 0$ and $b(\eta)<0$ that enable multistationarity and sets that do not enable multistationarity. 
For example, consider $\eta=(2,0.5,0.5,2,2,1,1,1)$, which belongs to case (i) of Theorem~\ref{thm:main}. Then $\kk=(  2, 2, 2 , 0.25, 0.25, 1 , 0.25, 0.25 , 1 , 1 ,1  , 1)$ is such that $\pi(\kk)=\eta$ and hence it does not enable multistationarity.

In order to construct a set that enables multistationarity, we consider any choice of $K_2,K_3,K_4$, $ \kk_3,\kk_6,\kk_9,\kk_{12}$ and consider the polynomial 
$g_{\eta'}(y)$. Letting $K_1=y^3$, we are guaranteed that for $y$ large enough, $p_\eta(y,1,y^2)$ is negative. The value of $K_1$ found in this way gives a choice of reaction rate constants that enable multistationarity. For example, let 
 $(K_2,K_3,K_4,\kk_3,\kk_6,\kk_9,\kk_{12})= (1,1,1,2,1,1,1)$, such that $a(\eta)>0$. Then 
\[g_{\eta'}(y)=  -2 y^3 +11 y^2 + 21y + 17. \]
By letting $y =8$ (or any $y$ larger than $\sim\,7.14$), $p_\eta(x)<0$ for $K_1=y^3=512$ and $x=(y,1,y^2)=(8,1,64)$. 
The choice $\kk=(1, 510, 2, 2,1,1,2,1,1,2,1,1)$ enables multistationarity, since  $\pi(\kk)=\eta$. 
In order to find a linear invariant subspace with multiple steady states, we compute
$x^*=\varphi(8,1,64)=(8,1,64,2,16,1,2,16,16)$ in \eqref{eq:parametrization}, which belongs to the linear invariant  
 subspace defined by $E_{\rm tot } = 25$, $F_{\rm tot}=19$, $S_{\rm tot}= 117$. We solve the equations for the positive steady states in this linear invariant subspace, and obtain $x^*$ together 
with   two other positive steady states, given approximately by: 
\begin{align*} 
(11.2, 3.4, 82.9, 1.1, 3.6, 1.8, 3.6, 12, 12), \ (7.1, 0.5, 44.3, 2.4, 33.9, 0.6, 1.2, 17.3, 17.3)
 \end{align*}
(and there are two other solutions, with negative components).

For $(K_2,K_3, \kk_3,\kk_6,\kk_9,\kk_{12})=(1,1,2,1,1,1)$,   Figure~\ref{fig:regions} shows  in the $(K_1,K_4)$-plane  the region 
where multistationarity is enabled by Theorem~\ref{thm:main}(ii) (in yellow), and the regions where it is not enabled from Theorem~\ref{thm:main}(i) (in blue).
We show as well the region where multistationarity is not enabled due to $b(\eta)>0$, which conforms a small portion of the blue region.

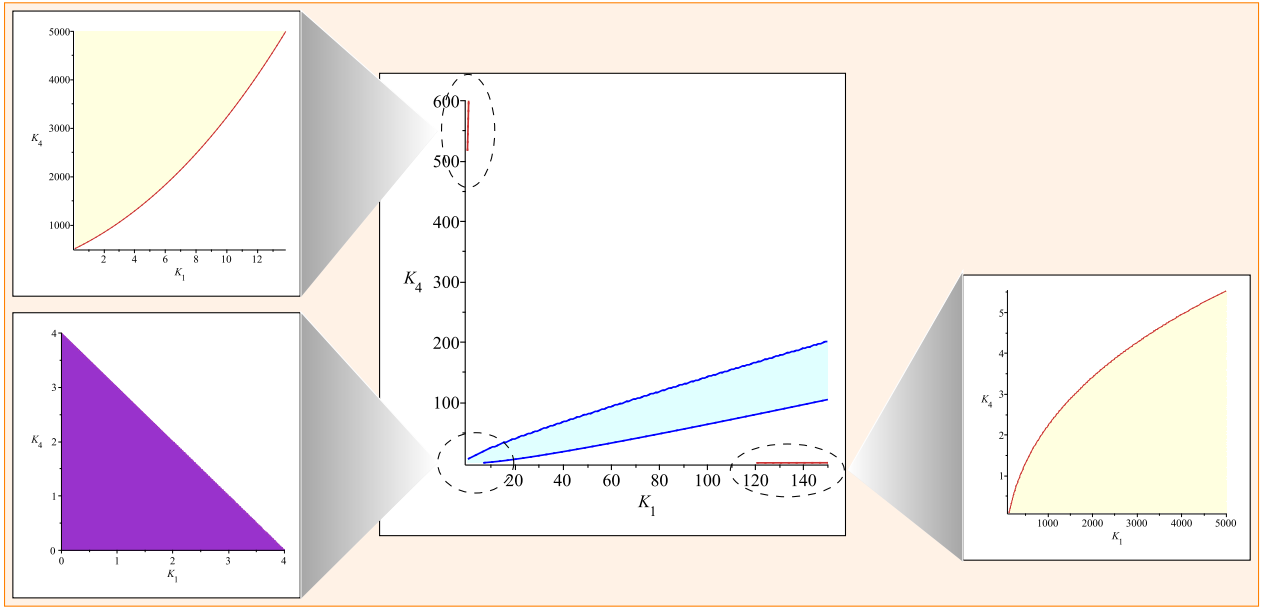
\begin{figure}[t]

\begin{tikzpicture}
\draw[fill=orange!10!white,draw=orange] (-2,5.5) rectangle (14.5,-2.5);
\node[draw=black,fill=white] (a) at (0,3.5) {
\includegraphics[scale=0.18]{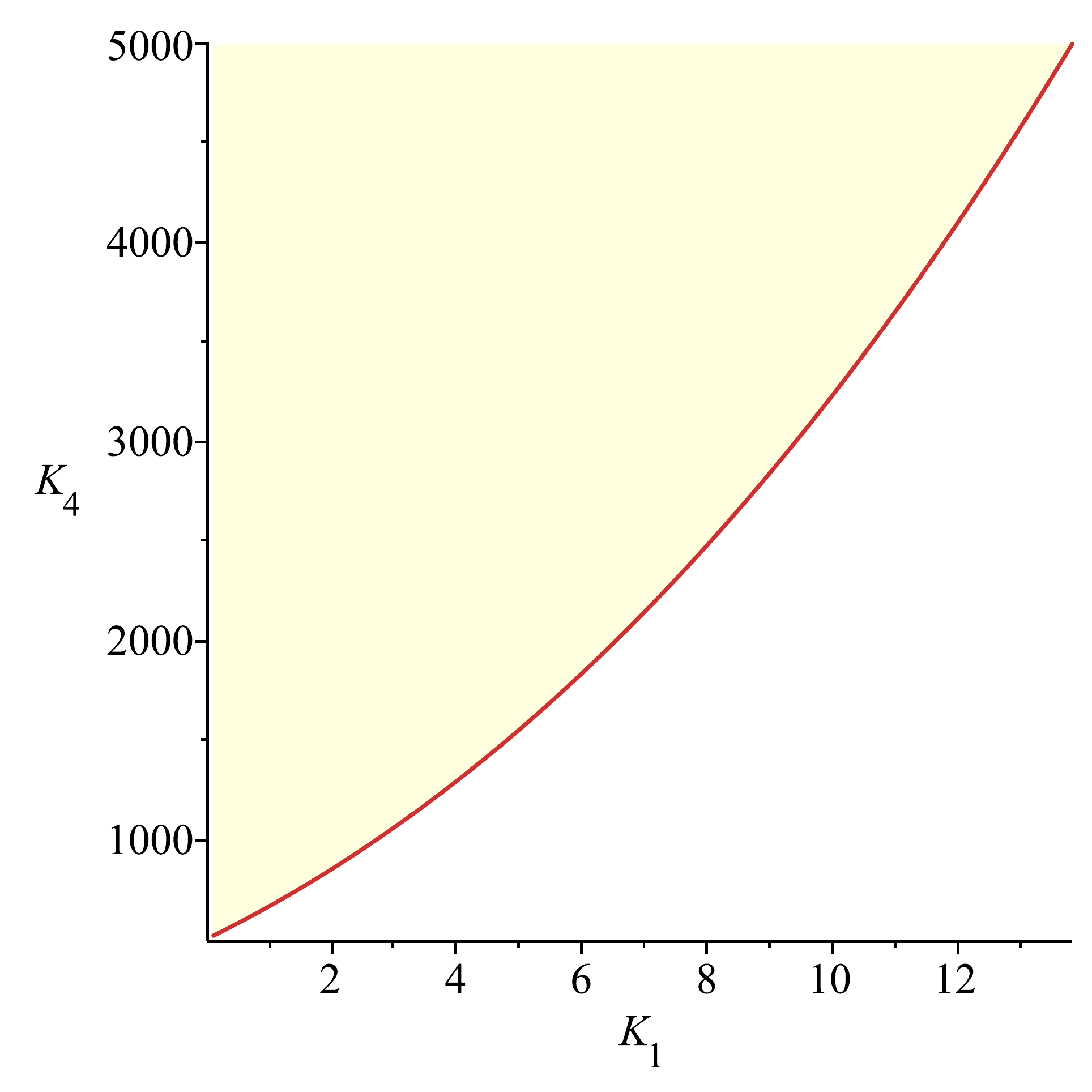}};
\node[draw=black,fill=white] (a) at (0,-0.5) {
\includegraphics[scale=0.18]{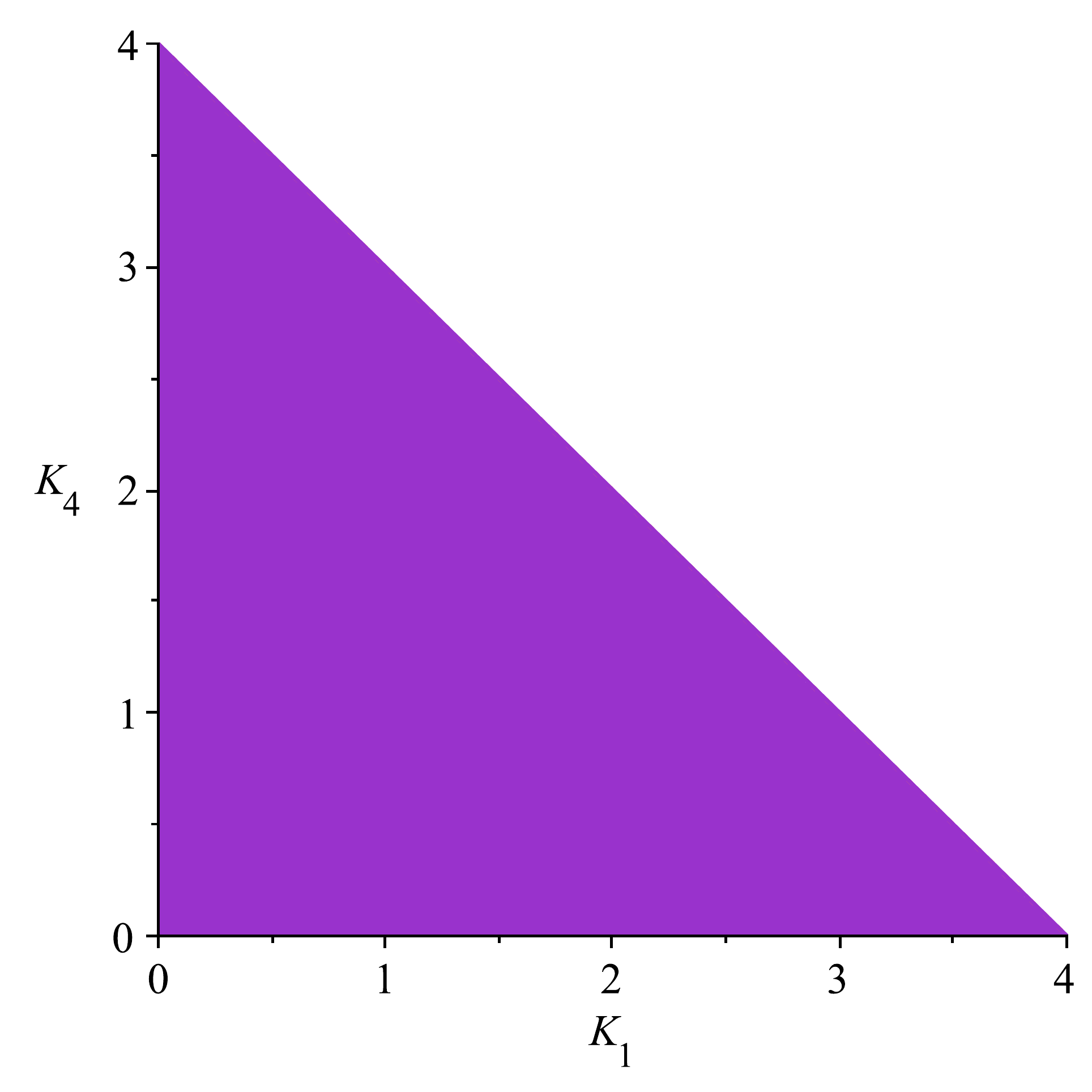}};
\node[draw=black,fill=white] (a) at (6,1.5) {
\includegraphics[scale=0.3]{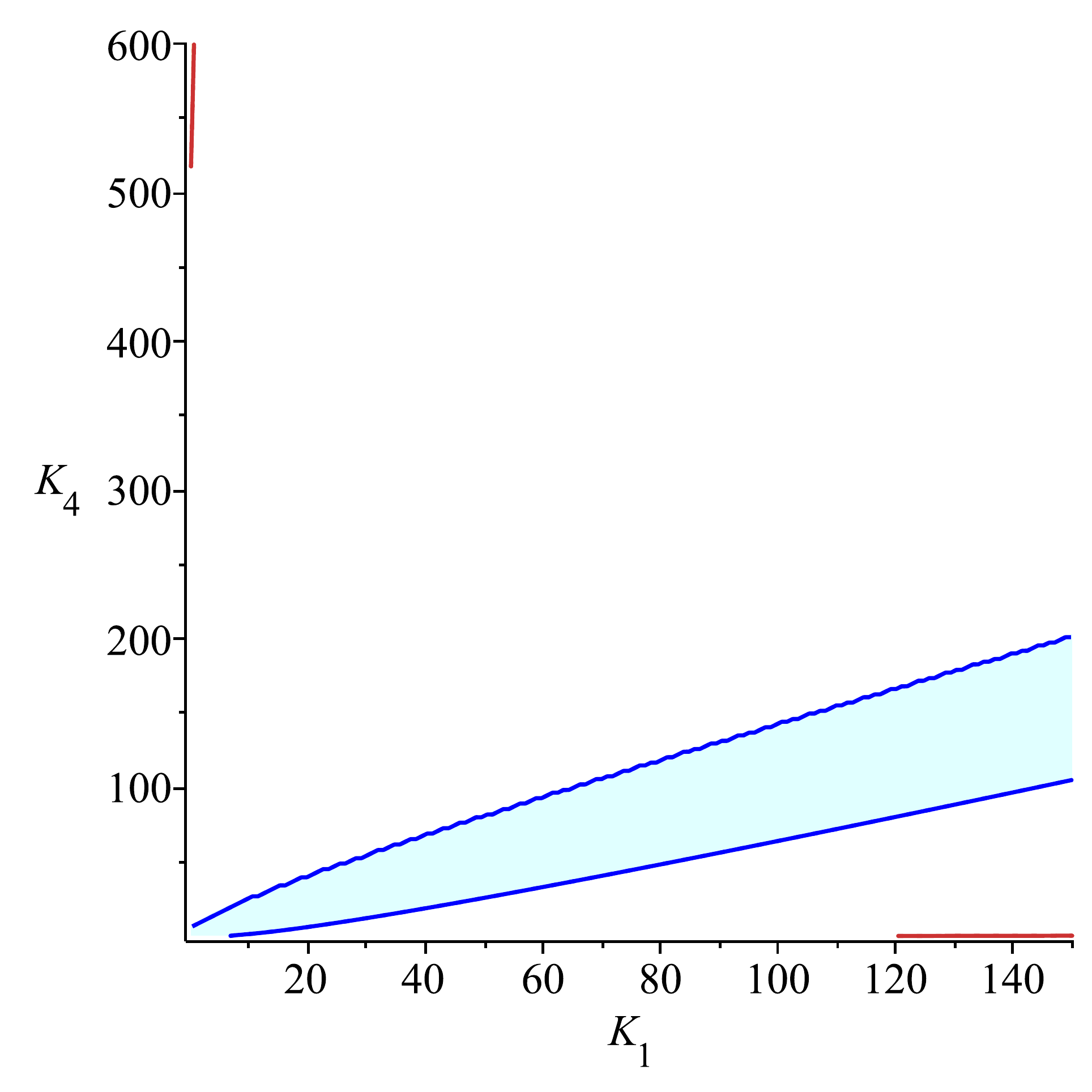}};
\node[draw=black,fill=white] (a) at (12.5,0) {
\includegraphics[scale=0.18]{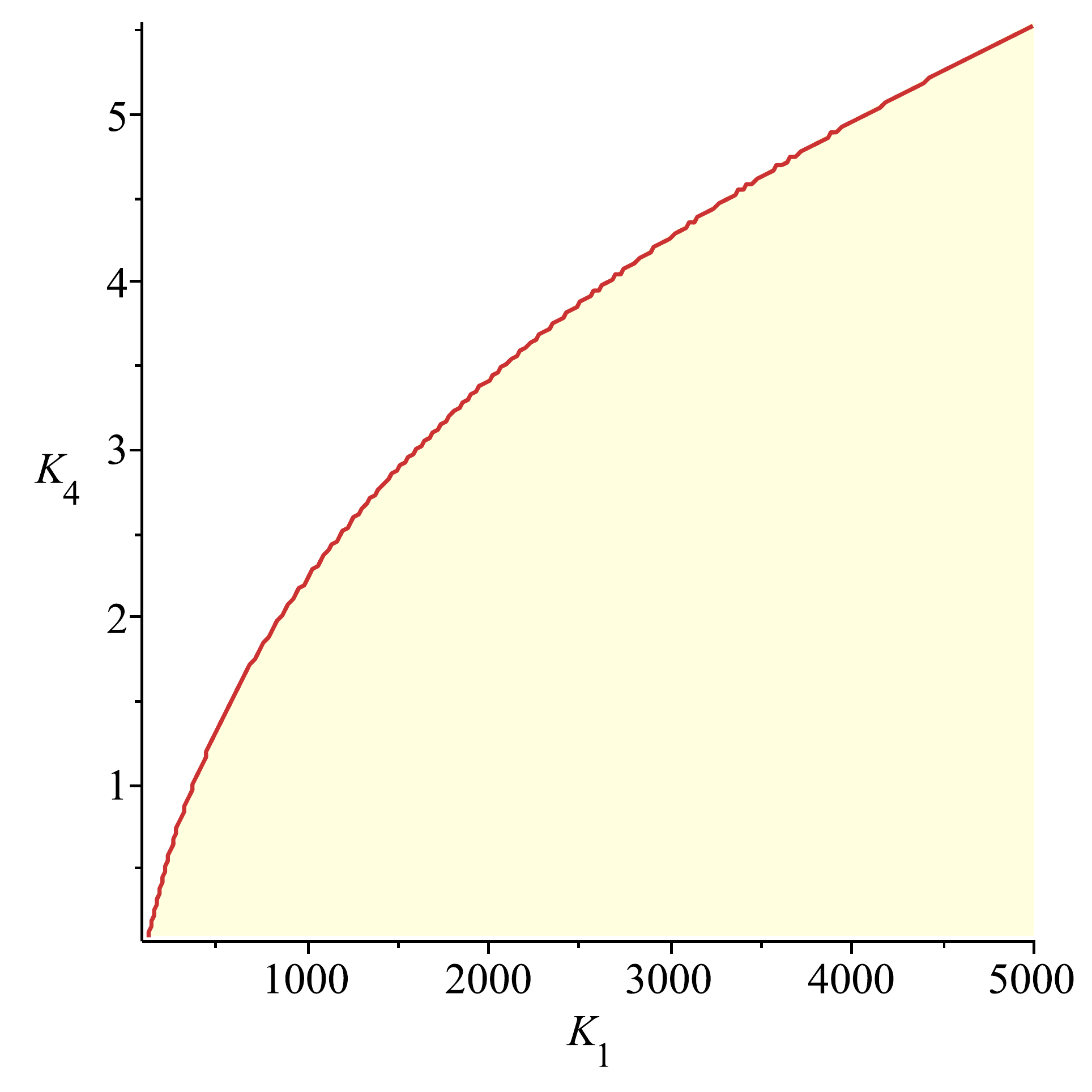}};
\draw[dashed] (4.1,3.8) ellipse (.35 and 0.75);
\draw[fill=gray!20!white,shade,shading angle=90,left color=gray!70!white,right color=white,draw=white] (3.7,3.8) -- (1.9,5.4) -- (1.9,1.6);
\draw[dashed] (8.3,-0.7) ellipse (.75 and 0.35);
\draw[fill=gray!20!white,shade,shading angle=90,right color=gray!70!white,left color=white,draw=white] (9.1,-0.7) -- (10.61,1.95) -- (10.61,-1.91);
\draw[dashed] (4.2,-0.6) ellipse (.5 and 0.4);
\draw[fill=gray!20!white,shade,shading angle=90,left color=gray!70!white,right color=white,draw=white] (3.7,-0.6) -- (1.9,1.4) -- (1.9,-2.4);
\end{tikzpicture}

\caption{{\small With $(K_2,K_3, \kk_3,\kk_6,\kk_9,\kk_{12})=(1,1,2,1,1,1)$, the figure shows in blue the values of  $K_1,K_4$ that do not enable multistationarity according to Theorem~\ref{thm:main}(i), and in yellow the region enabling multistationarity according to Theorem~\ref{thm:main}(ii). Two close ups of the multistationarity regions are given. In purple we show the region known from \cite{conradi-mincheva}.}}\label{fig:regions}
\end{figure}

Figure~\ref{fig:regions} illustrates that  Theorem~\ref{thm:main} leaves large regions undetermined. In order to inspect this more closely for this choice of parameters, we have sampled uniformly values of $K_1,K_4$. For each point, we have tested whether the polynomial $p_{\eta}(x)$ is negative for some positive $x$, see Figure~\ref{fig:sampling}. We have used the package RegularChains of Maple 2018 to this end. 
Points that do not enable multistationarity, that is, such that $p_\eta(x)>0$ for all $x\in \R^3_{>0}$, are shown in blue, while points that enable multistationarity, that is, $p_\eta(x)<0$ for some $x\in \R^3_{>0}$,  are shown in orange. The solid orange line shows the boundary of the multistationary region in Theorem~\ref{thm:main}(ii). The figure suggests that we should expect the region of multistationarity in the $(K_1,K_4)$-plane to be concentrated around the $K_1$- and $K_4$-axes after enlarging the region given in  Theorem~\ref{thm:main}(ii).

The green line in  Figure~\ref{fig:sampling} arises from a different choice of separating hyperplane in the proof of Theorem~\ref{thm:main}(ii). In the proof, we consider the vector $\omega=(3,1,0,2)$. The vector 
$\omega = (5,2,0,3)$ defines also a separating hyperplane, which in turn gives rise to a new function $g_{\eta'}(x)$. The figure illustrates that, the best choice of separating hyperplane depends on the values of $K_1$ and $K_4$. In general, there are infinitely many separating hyperplanes that can be explored. 

\begin{figure}[t]
\begin{tikzpicture}
\draw[fill=orange!10!white,draw=orange] (-3,2.9) rectangle (14,-2.9);
\node[draw=black,fill=white] (a) at (0,0) {
\includegraphics[scale=0.25]{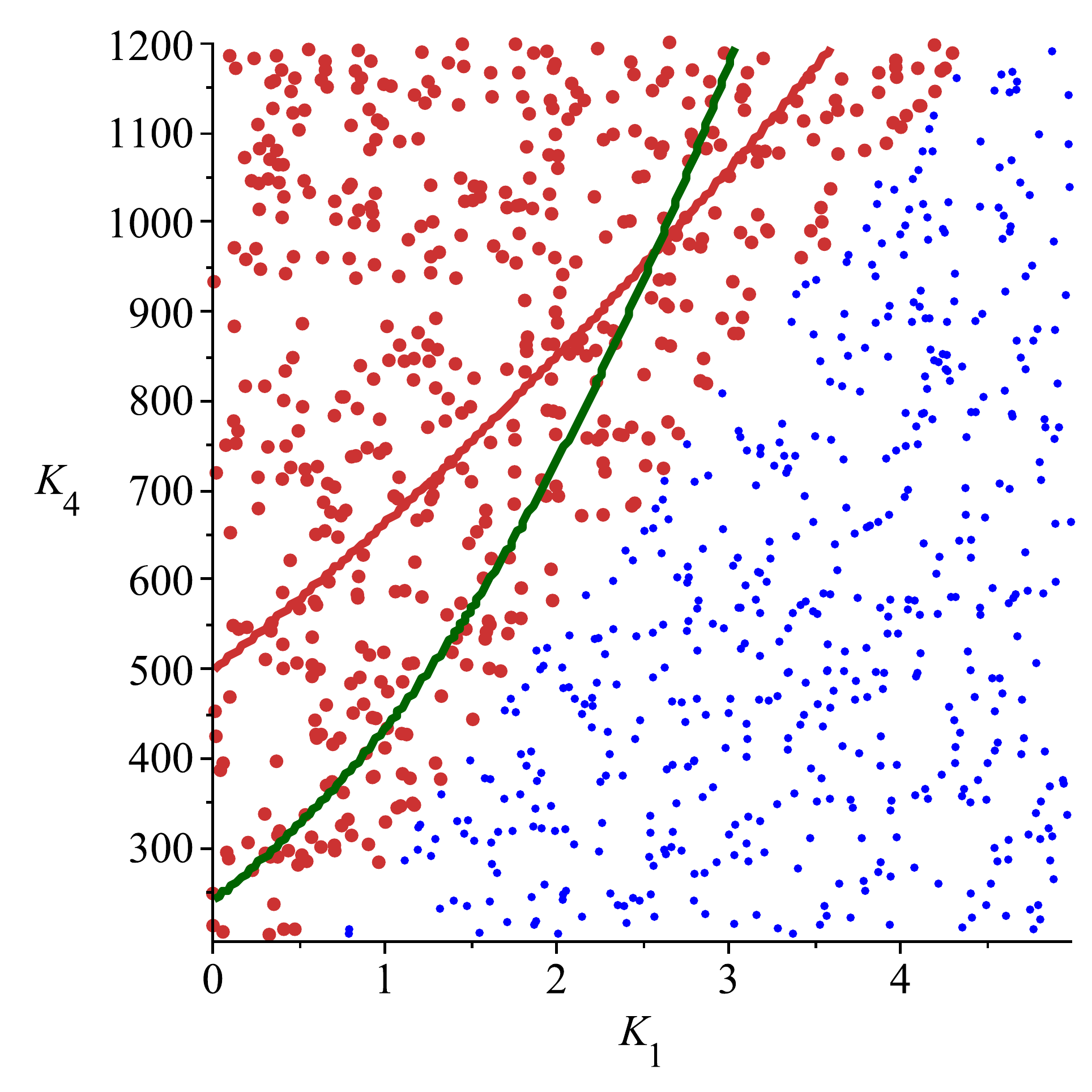}};
\node[draw=black,fill=white] (a) at (5.5,0) {
\includegraphics[scale=0.25]{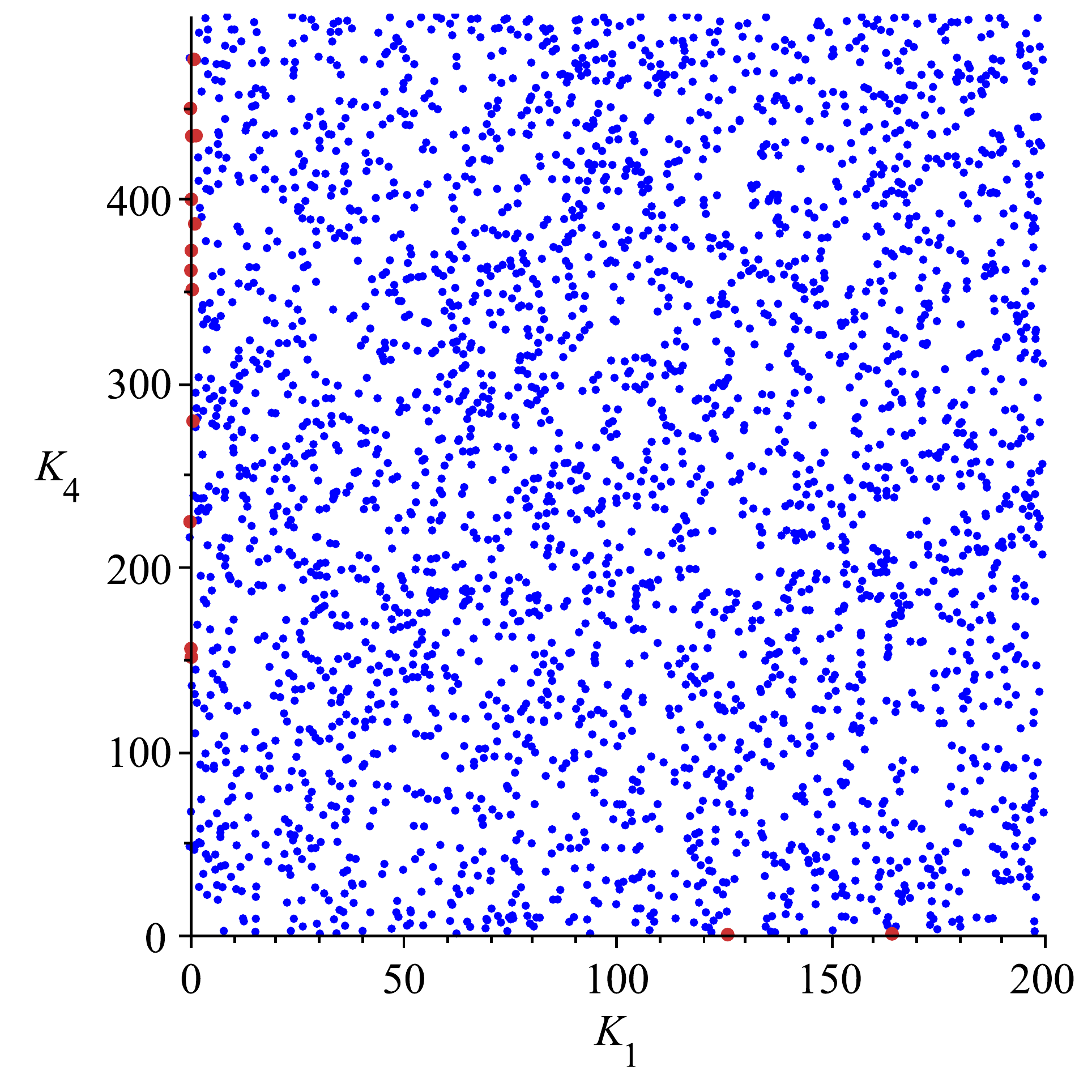}};
\node[draw=black,fill=white] (a) at (11,0) {
\includegraphics[scale=0.25]{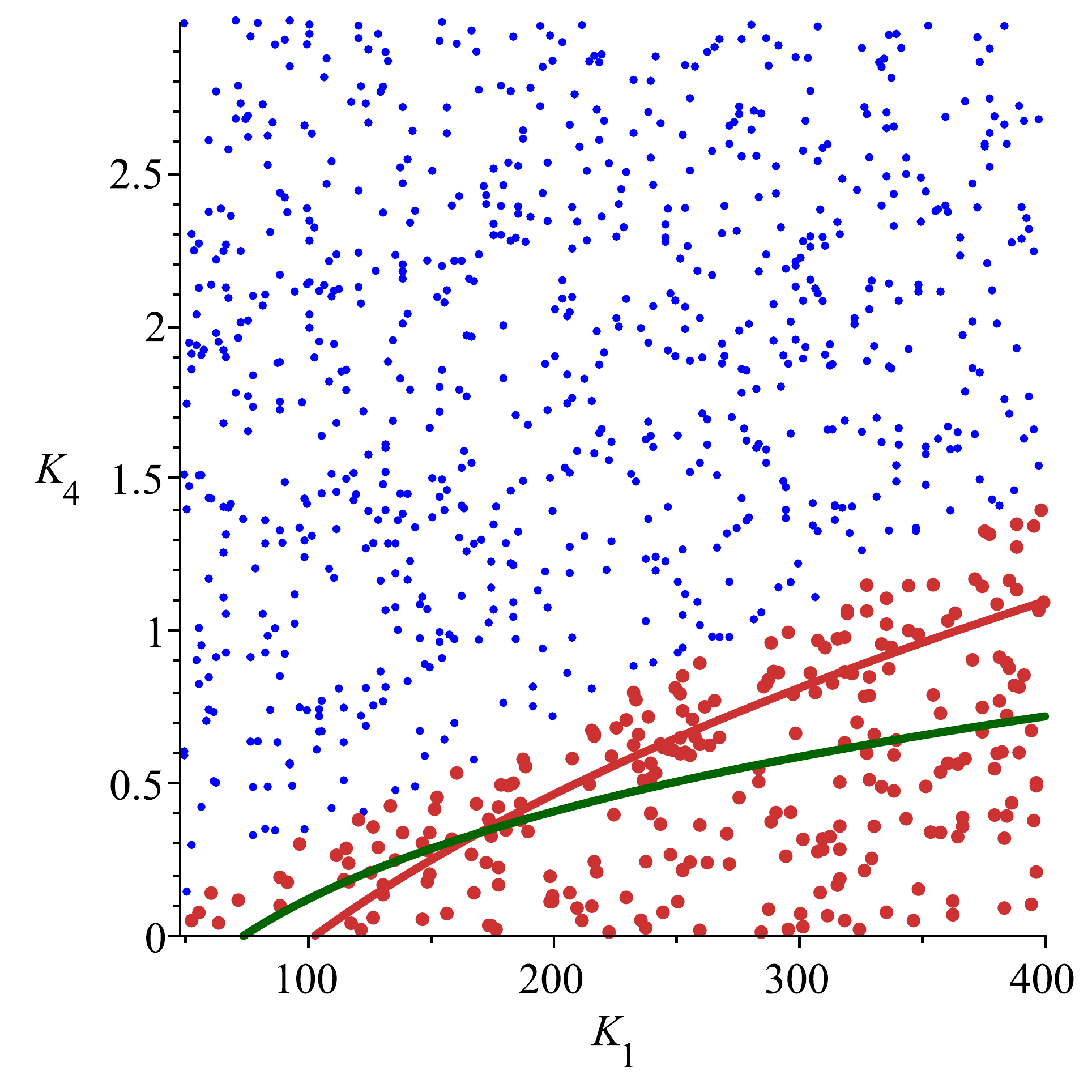}};
\end{tikzpicture}

\caption{{\small We consider $(K_2,K_3, \kk_3,\kk_6,\kk_9,\kk_{12})=(1,1,2,1,1,1)$. A blue point indicates the choice of $K_1,K_4$ does not enable multistationarity, while an orange point indicates it does. The middle figure gives an overview, and the two side figures show a close up around the $K_1$- and $K_4$-axes. The orange line is the boundary of the multistationarity region in Theorem~\ref{thm:main}(ii), and the green line is another boundary for multistationarity found using another separating hyperplane in the proof of Theorem~\ref{thm:main}(ii) (see text).}}\label{fig:sampling}

\end{figure}

 \section{Closing remarks}
In this work we have contributed to understanding the region of reaction rate constants that enable multistationarity for the model model, namely the two-site phosphorylation cycle. It is the hope that the techniques used here might be applicable to study other situations in which the sign of a multivariate polynomial on the positive orthant needs to be determined. For instance, the allosteric kinase model given in \cite{feng:allosteric} presents the same difficulties as the model model.

Furthermore, the study of the signs of a polynomial is also key when analysing stability of steady states or the presence of Hopf bifurcations (see for example \cite{torres:stability,shiu:hopf}).

\subsection*{Acknowledgements} EF   acknowledges funding from the Danish Research Council for Independent Research and thanks Carsten Wiuf for comments on the manuscript.

\small 

 \end{document}